\documentclass{svjour3}                     
\smartqed  
\usepackage{graphicx}
\usepackage{float}
\usepackage{mathptmx}      
\usepackage{latexsym}
\usepackage{amsmath}
\usepackage{amsfonts}
\usepackage[psamsfonts]{amssymb}
\newtheorem{observation}{Observation}
\newcommand{\outlet}{{\mathop{\rm outlet}}}
\newcommand{\support}{{\mathop{\rm support}}}
\newcommand{\crib}{{\mathop{\rm crib}}}
\newcommand{\unconf}{{\mathop{\rm unconf}}}
\newcommand{\tw}{{\mathop{\rm tw}}}
\newcommand{\hide}[1]{}

\newcommand{\BB}{{\mathcal B}}
\newcommand{\CC}{{\mathcal C}}
\newcommand{\II}{{\mathcal I}}

\newcommand{\TT}{{\mathcal T}}
\newcommand{\OO}{{\mathcal O}}
\newcommand{\PP}{{\mathcal P}}

\newcommand{\Ss}{{\mathcal S}}

\newcommand{\XX}{{\mathcal X}}

\begin{document}

\title{Positive-instance driven dynamic programming for treewidth
\footnote{A preliminary and an abridged version of this paper
was presented at the 25th European Sysmposium on Algorithms}} 


\author{Hisao Tamaki}


\institute{
H. Tamaki
\at
Kawasaki, 214-8571, Japan, 
Meiji University,
1-1-1 Higashi-Mata, Tama,  
Kawasaki, 214-8571, Japan, 
              \\
              Tel.: +81-44-934-7478\\
              \email{tamaki@cs.meiji.ac.jp}
}

\date{Received: date / Accepted: date}

\maketitle

\begin{abstract}
Consider a dynamic programming scheme for a decision problem
in which all subproblems involved 
are also decision problems. An implementation of
such a scheme is {\em positive-instance driven} (PID), 
if it generates positive subproblem instances, but not
negative ones, building each on smaller positive instances.

We take the dynamic programming scheme
due to Bouchitt\'{e} and Todinca for treewidth computation,
which is based on minimal separators and potential maximal cliques,
and design a variant (for the decision version of the problem) 
with a natural PID implementation.
The resulting algorithm performs extremely well:
it solves a number of standard benchmark instances for which
the optimal solutions have not previously been known.
Incorporating a new heuristic algorithm for detecting safe
separators, it also solves all of the 100 public instances posed by the exact
treewidth track in PACE 2017, a competition on algorithm implementation.

We describe the algorithm, prove its correctness,
and give a running time bound 
in terms of the number of positive subproblem instances.
We perform an experimental analysis which supports
the practical importance of such a bound.

\end{abstract}

\section{Introduction}
Suppose we design a dynamic programming algorithm for some decision problem,
formulating subproblems, which are decision problems as well, and recurrences
among those subproblems. A standard approach is to list all subproblem
instances and scan the list from ``small" ones to ``large" , deciding the
answer, positive or negative, to each instance by means of these recurrences.
When the number of positive subproblem instances 
are expected to be much smaller than the total number of subproblem instances, 
a natural alternative is to generate positive instances only, using recurrences
to combine positive instances to generate a ``larger"
positive instance.  We call such a mode of dynamic programming execution
{\em positive-instance driven} or {\em PID} for short.  One goal of
this paper is to demonstrate that PID is not simply a
low-level implementation strategy but can be a paradigm 
of algorithm design for some problems.

The decision problem we consider is that of deciding,
given graph $G$ and positive integer $k$, if the
treewidth of $G$ is at most $k$.  This graph parameter
was introduced by Robertson and Seymour \cite{RS86} and
has had a tremendous impact on graph theory and on the
design of graph algorithms (see, for example, a survey \cite{BK08}).
The treewidth problem is NP-complete \cite{ACP87} but fixed-parameter tractable:
it has an $f(k)n^{O(1)}$ time algorithm for some fixed function $f(k)$ as
implied by the graph minor theorem of Robertson and Seymour \cite{RS04}, 
and an explicit
$O(f(k)n)$ time algorithm was given by Bodlaender \cite{Bod96}.
A classical dynamic programming algorithm due to Arnborg, Corneil, and
Proskurowsky (ACP algorithm) \cite{ACP87} runs in $n^{k + O(1)}$ time.
Bouchitt\'{e} and Todinca \cite{BT02} developed a more refined
dynamic programming algorithm (BT algorithm) 
based on the notions of 
minimal separators and potential maximal cliques, which lead
to algorithms running in $O(1.7549^n)$ time or in 
$O(n^5 \tbinom{\lceil (2n + k + 8)/3 \rceil}{k + 2})$ time
\cite{FKTV08,FV12}.
Another important approach to treewidth computation is
based on the perfect elimination order (PEO) of minimal chordal
completions of the given graph. 
PEO-based dynamic programming 
algorithms run in $O^*(2^n)$ time with exponential space and  
in $O^*(4^n)$ time with polynomial space \cite{BFKKT12}, where
$O^*(f(n))$ means $O(n^c f(n))$ for some constant $c$.

There has been a considerable amount of effort on implementing
treewidth algorithms to be used in practice and, prior to this work,
the most successful implementations 
for exact treewidth computation are all based on PEO.
The authors of \cite{BFKKT12} implemented the $O^*(2^n)$ time dynamic
programming algorithm and experimented on its performance, showing that it works well
for small instances. 
For larger instances, PEO-based branch-and-bound algorithms are known to work
well in practice \cite{GD04}.
Recent proposals for
reducing treewidth computation to SAT solving are also based on PEO
\cite{SH09,BJ14}.
From the PID perspective, this situation is somewhat surprising,
since it can be shown that each positive subproblem instance
in the PEO-based dynamic programming scheme corresponds to a combination
of an indefinite number of positive subproblem instances in the
ACP algorithm, and hence the number of positive subproblem
instances can be exponentially larger than that in the ACP algorithm. 
Indeed, a PID variant of the ACP 
algorithm was implemented by the present author and has won the first place in
the exact treewidth track of PACE 2016 \cite{DHJKKR17}, 
a competition on algorithm implementations, outperforming
other submissions based on PEO.  Given this success, a natural next step
is to design a PID variant of the BT algorithm, which is tackled 
in this paper.

The resulting algorithm performs extremely well, 
as reported in Section~\ref{sec:performance}.
It is tested on DIMACS graph-coloring instances \cite{JT93}, which have been
used in the literature on treewidth computation as standard benchmark instances 
\cite{GD04,BK06,Musliu08,SH09,BFKKT12,BJ14}. 
Our implementation of the algorithm solves all the instances
that have been previously solved (that is, with 
matching upper and lower bounds known) within 10
seconds per instance on a typical desktop computer and solves 13 out of the
42 previously unsolved instances.  For nearly half of the instances
which it leaves unsolved, it significantly reduces the gap between 
the lower and upper bounds. It is interesting to
note that this is done by improving the lower bound. 
Since the number of positive subproblem instances are much smaller when 
$k < \tw(G)$ than when $k = \tw(G)$, 
the PID approach is particularly good at establishing strong lower bounds.

We also adopt the notion of safe separators due to Bodlaender and
Koster \cite{BK06} in our preprocessing and design a new heuristic
algorithm for detecting safe separators. 
With this preprocessing, our implementation also solves
all of the 100 public instances posed by PACE 2017 \cite{PACE17}, 
the successor of PACE 2016.
It should be noted that
these test instances of PACE 2017 are much harder than those of
PACE 2016: the winning implementation of PACE 2016 mentioned above,
which solved 199 of the 200 instances therein, 
solves only 62 of these 100 instances of PACE 2017  
in the given time of 30 minutes per instance.

Adapting the BT algorithm to work in PID mode has turned out non-trivial.
Each subproblem instance in the BT algorithm for given graph $G$ and
positive integer $k$ takes the form of
a connected set $C$ of $G$ such that $N_G(C)$, 
the open neighborhood of $C$ in $G$,
is a minimal separator of $G$ with cardinality at most $k$.
For each such $C$, we ask if $C$ is {\em feasible}, in the
sense that there is a tree decomposition of the subgraph of
$G$ induced by $C \cup N_G(C)$ of width at most $k$ that has a bag
containing $N_G(C)$ (see Section~\ref{sec:prelim} for the
definition of a tree decomposition of a graph).
The difficulty of making the BT algorithm PID comes from the
fact that the recurrence for deciding if $C$ is feasible 
may involve an indefinite number of connected sets $C'$
such that $C' \subset C$. Thus, even if the number of
positive instances is small, there is a possibility that
the running time is exponential in that number.
We approach this issue by introducing an auxiliary structure
we call O-blocks (see Section~\ref{sec:oriented}) and
formulate a recurrences that are binary: a combination of
a feasible connected set and a feasible O-block may yield either
a larger feasible connected set or a larger feasible O-block.
Due to this binary recurrence, we obtain an upper bound
on the running time of our algorithm which is sensitive to
the number of subproblem instances 
(Observation~\ref{obs:run_time} in Section~\ref{sec:time}).
To support the significance of such a bound, we perform an
experimental analysis which shows the existence of huge gaps
between the actual number of combinatorial objects corresponding
to subproblems and the known theoretical upper bounds.

The rest of this paper is organized as follows.
In Section~\ref{sec:prelim}, we introduce notation, define basic
concepts and review facts in the literature. 
In Section~\ref{sec:oriented}, we precisely define the subproblems
in our dynamic programming algorithm and formulate recurrences.
We describe our algorithm and
prove its correctness in Section~\ref{sec:algorithm} and then 
analyze its running time in Section~\ref{sec:time}.
In Section~\ref{sec:experimental}, 
we describe our experimental analysis.
In Section~\ref{sec:implementation}, we describe some implementation
details. Finally, in Section~\ref{sec:performance}, we give details
of the performance results sketched above.

\section{Preliminaries} \label{sec:prelim}
In this paper, all graphs are simple, that is, without self loops or
parallel edges. Let $G$ be a graph. We denote by $V(G)$ the vertex set
of $G$ and by $E(G)$ the edge set of $G$.
For each $v \in V(G)$, $N_G(v)$ denote the set of neighbors of $v$ in $G$:
$N_G(v) = \{u \in V(G) \mid \{u, v\} \in E(G)$.
For $U \subseteq V(G)$, the {\em open neighborhood of $U$ in $G$}, denoted
by $N_G(U)$,  is the set of vertices adjacent to some vertex in $U$ but not
belonging to $U$ itself: $N_G(U) = (\bigcup_{v \in U} N_G(v)) \setminus U$.
The {\em closed neighborhood of $U$ in $G$}, denoted by $N_G[U]$, is defined
by $N_G[U] = U \cup N_G(U)$.  We also write $N_G[v]$ for $N_G[\{v\}] = N_G(v)
\cup \{v\}$. We denote by $G[U]$ the subgraph of $G$ induced by $U$: 
$V(G[U]) = U$ and $E(G[U]) = \{\{u, v\} \in E(G) \mid u, v \in U\}$.
In the above notation, as well as in the notation further introduced below,
we will often drop the subscript $G$ when the graph is clear from the context. 

We say that vertex set $C \subseteq V(G)$ is {\em connected in}
$G$ if, for every $u, v \in C$, there is a path in $G[C]$ between $u$
and $v$. It is a {\em connected component} of $G$ if it is connected
and is inclusion-wise maximal subject to this condition.
A vertex set $C$ in $G$
is a {\em component associated with $S \subseteq G$}, if
$C$ is a connected component of $G[V(G) \setminus S]$.
For each $S \subseteq V(G)$, we denote by
$\CC_G(S)$
the set of all components associated with $S$.
A vertex set $S \subseteq V(G)$ is a {\em separator} of $G$ if 
$|\CC_G(S)| > {\CC_G(\emptyset})|$, that is, if its removal
increases the number of connected components of $G$.  
A component $C$ associated with separator $S$ of $G$ is
a {\em full component} if $N_G(C) = S$.
A separator $S$ is a {\em minimal separator} 
if there are at least two full components associated with $S$.
This term is justified by this fact:
if $S$ is a minimal separator and $a$, $b$ vertices belonging to 
two distinct full components associated with $S$, 
then for every proper subset $S'$ of $S$, $a$ and $b$ belong to the
same component associated with $S'$; $S$ is a minimal set of vertices that
separates $a$ from $b$. A {\em block} is a pair $(S, C)$,
where $S$ is a separator and $C$ is a component associated with
$S$; it is a {\em full block} if $C$ is a full component, 
that is, $S = N(C)$.

Graph $H$ is {\em chordal} if every induced cycle of $H$ has length exactly
three. $H$ is a {\em minimal chordal completion of $G$} if it is chordal,
$V(H) = V(G)$, $E(G) \subseteq E(H)$, and $E(H)$ is minimal subject to these
conditions. A vertex set $\Omega \subseteq V(G)$ is a {\em potential maximal}
clique of $G$, if $\Omega$ is a clique in some minimal chordal completion of $G$.

A {\em tree-decomposition} of $G$ is a pair $(T, \XX)$ where $T$ is a tree
and $\XX$ is a family $\{X_i\}_{i \in V(T)}$ of vertex sets of $G$ such that
the following three conditions are satisfied. We call members of
$V(T)$ {\em nodes} of $T$ and each $X_i$ the {\em bag} at node $i$.   
\begin{enumerate}
  \item $\bigcup_{i \in V(T)} X_i = V(G)$.
  \item For each edge $\{u, v\} \in E(G)$, there is some $i \in V(T)$
  such that $u, v \in X_i$.
  \item The set of nodes $I_v = \{i \in V(T) \mid v \in X_i\}$ of
  $V(T)$ induces a connected subtree of $T$.
\end{enumerate}
The {\em width} of this tree-decomposition is $\max_{i \in V(T)} |X_i| - 1$.
The {\em treewidth} of $G$, denoted by $\tw(G)$ is the minimum width
of all tree-decompositions of $G$.
We may assume that the bags $X_i$ and $X_j$ are distinct from each other
for $i \neq j$ and, under this assumption, we will often regard
a tree-decomposition as a tree $T$ in which each node is a bag.

We call a tree-decomposition $T$ of $G$ {\em canonical} if
each bag of $T$ is a potential maximal clique of $G$ and,
for every pair $X$, $Y$ of adjacent bags in $T$, 
$X \cap Y$ is a minimal separator of $G$.
The following fact is well-known. 
It easily follows, for example, from
Proposition~2.4 in \cite{BT01}.
\begin{lemma} 
\label{lem:pmc_decompose}
Let $G$ be an arbitrary graph.
There is a tree-decomposition $T$ of $G$ of 
width $\tw(G)$ that is canonical.
\end{lemma}

The following local characterization of a potential maximal clique
is crucial.  We say that a vertex set $S \subseteq V(G)$ is {\em cliquish}
in $G$ if, for every pair of distinct vertices $u$ and $v$ in $S$,
either $u$ and $v$ are adjacent to each other or there is some $C \in \CC(S)$
such that $u, v \in N(C)$.  In other words, $S$ is cliquish if
completing $N(C)$ for every $C \in \CC(S)$ into a clique makes $S$ a clique.  
\begin{lemma}
\label{lem:pmc_charact}
(Theorem 3.15 in \cite{BT01})
A separator $S$ of $G$ is a potential maximal clique
of $G$ if and only if (1) $S$ has no full-component
associated with it and (2) $S$ is cliquish.
\end{lemma}

It is also shown in \cite{BT01} that
if $\Omega$ is a potential maximal clique
of $G$ and $S$ is a minimal
separator contained in $\Omega$, then there
is a unique component $C_S$ associated with $S$
that contains $\Omega \setminus S$. We need an
explicit way of forming $C_S$ from $\Omega$ and $S$.

Let $K \subseteq V(G)$ be an arbitrary vertex set
and $S$ an arbitrary proper subset of $K$.
We say that a component $C \in \CC(K)$ is
{\em confined to $S$} if $N(C) \subseteq S$; otherwise it is
{\em unconfined to $S$}.
Let $\unconf(S, K)$ denote the set of
components associated with $K$ that are unconfined to $S$.
Define the {\em crib} of $S$ with respect to $K$, 
denoted by $\crib(S, K)$, to be $(K \setminus S)
\cup \bigcup_{C \in \unconf(S, K)} C$: it is the
union of $K \setminus S$ and 
all those components associated with $K$ that have
neighborhoods intersecting $K \setminus S$. 

The following lemma relies only on the second property of potential 
maximal cliques, namely that they are cliquish, 
and will be applied not only to potential maximal cliques but also
to separators with full components, which are trivially cliquish.

\begin{lemma}
\label{lem:crib}
Let $K \subseteq V(G)$ be a cliquish vertex set.
Let $S$ be an arbitrary proper subset of $K$.
Then, $\crib(S, K)$ is a full component
associated with $S$. 
\end{lemma}
\begin{proof}
Let $C = \crib(S, K)$.
We first show that $G[C]$ is connected. 
Suppose $K \setminus S$ has two distinct vertices $u$ and $v$.
Since $K$ is cliquish, either
$u$ and $v$ are adjacent to each other or there is some component
$C' \in \CC(K)$ such that $u, v \in N(C')$.
In the latter case, as $C'$ is unconfined to $S$, 
we have $C' \subseteq C$. 
Therefore, $u$ and 
$v$ belong to the same connected component of $G[C]$.
As this applies to every pair of vertices in 
$K \setminus S$, $K \setminus S$ is contained 
in a single connected component of $G[C]$.
Moreover, each component $C' \in \CC(K)$ contained
in $C$ is unconfined to $S$, by the definition of
$\crib(S, K)$, and hence has a neighbor in 
$K \setminus S$.
Therefore, we conclude that $G[C]$ is connected.
Each vertex $v$ not in $S \cup C$ belongs to
some component in $\CC(K)$ that is confined to
$S$ and hence does not have a neighbor in $C$.
Therefore, $C$ is a component associated with $S$.

To see that $C$ is a full component, 
let $u \in S$ and $v \in K \setminus S$ be arbitrary. 
Since $K$ is cliquish, 
either $u$ and $v$ are adjacent to each other
or there is some $C' \in \CC(K)$ such that 
$u, v \in N(C')$. As such $C'$ is unconfined to $S$
in the latter case, 
we conclude that $u \in N(C)$ in either case. 
Since this holds 
for arbitrary $u \in S$, we conclude that
$C$ is a full component associated with $S$.
\qed 
\end{proof}
\begin{remark}
As $\crib(S, K)$ contains $K \setminus S$,
it is clear that it is the only component
associated with $S$ that intersects $K$.
Therefore, the above mentioned assertion on
potential maximal cliques is a corollary 
to this Lemma.
\end{remark}

\section{Recurrences on oriented minimal separators}
\label{sec:oriented}
In this section, we fix graph $G$ and positive integer $k$ that are
given in the problem instance: we are to decide if
the treewidth of $G$ is at most $k$. We assume that $G$
is connected.  

For connected set $C \subseteq V(G)$, we denote by
$G\langle C \rangle$ the graph obtained from $G[N[C]]$ by
completing $N(C)$ into a clique: $V(G\langle C \rangle) = N[C]$ and
$E(G\langle C \rangle) = E(G[N[C]]) \cup \{\{u, v\} \mid u,v \in N(C),
u \neq v\}$.  We say $C$ is {\em feasible} if
$\tw(G\langle C \rangle) \leq k$.
Equivalently,
$C$ is feasible if $G[N[C]]$ has a tree-decomposition
of width $k$ or smaller that has a bag
containing $N(C)$.

Let us first review the BT algorithm \cite{BT01} adapting
it to our decision problem. 
We first list all minimum separators of cardinality $k$ or smaller
and all potential maximal cliques of cardinality $k + 1$ or smaller.
Then, for each pair of a potential maximal clique $\Omega$ and
a minimal separator $S$ such that $S \subset \Omega$, place a link
from $S$ to $\Omega$. To understand the difficulty of
formulating a PID variant of the algorithm, 
it is important to note that the pair $(\Omega, S)$ to be linked
is easy to find from the side of $\Omega$, but not the other way round.
Then, we scan the full blocks $(N(C), C)$ of minimal separators
in the increasing order of $|C|$ to decide if $C$ is feasible, using
the following recurrence: $C$ is feasible if and only if
there is some potential maximal clique $\Omega$ such that
$N(C) \subset \Omega$, $C = \crib(N(C), \Omega)$, and
every component $D \in \unconf(N(C), \Omega)$ is feasible.
Finally, we have $\tw(G) \leq k$ if and only if
there is a potential maximal clique $\Omega$ with 
$|\Omega| \leq k + 1$ such that
every component associated with $\Omega$ is feasible.

To facilitate the PID construction, we orient minimal separators
as follows. We assume a total order $<$ on $V(G)$.
For each vertex set $U \subseteq V(G)$, 
the {\em minimum element} of $U$, denoted by $\min(U)$, is the smallest
element of $U$ under $<$.
For vertex sets $U$ and $W$, we say 
{\em $U$ precedes $W$} and write $U \prec W$ if $\min(U) < \min(W)$.

We say that a connected set $C$ is {\em inbound} if
there is some full block associated with
$N(C)$ that precedes $C$; otherwise, it is {\em outbound}. 
Observe that if $C$ is inbound then $N(C)$ is a minimal
separator, since $N(C)$ has another full component associated with it and,
contrapositively, if $N(C)$ is not a minimal separator then $C$ is
necessarily outbound.
We say a full block $(N(C), C)$ is {\em inbound}
({\em outbound}) if $C$ is inbound (outbound, respectively).

\begin{lemma}
\label{lem:pmc_outbounds}
Let $K$ be a cliquish vertex set 
and let $A_1, A_2$
be two components associated with $K$.
Suppose that $A_1$ and $A_2$ are outbound.
Then, either $N(A_1) \subseteq N(A_2)$ or 
$N(A_2) \subseteq N(A_1)$.
\end{lemma}
\begin{proof}
Let $K$, $A_1$, and $A_2$ be as above and suppose neither of 
$N(A_1)$ and $N(A_2)$ is a subset of the other. 
For $i = 1, 2$, let $C_i = \crib(N(A_i), K)$.
Since $N(A_2) \setminus N(A_1)$ is non-empty and 
contained in $K \setminus N(A_1)$, 
$A_2$ is contained in $C_1$. We have 
$A_1 \prec C_1$ as $A_1$ is
outbound and hence $A_1 \prec A_2$.  
A contradiction, 
since similarly we have $A_2 \prec A_1$.
\qed 
\end{proof}

Let $K$ be a cliquish vertex set. Based on the
above lemma, we define the {\em outlet}
of $K$, denoted by $\outlet(K)$, as follows.
If no non-full component associated with $K$ is outbound,
then we let $\outlet(K) = \emptyset$. Otherwise, 
$\outlet(K) = N(A)$, where $A$ is a non-full component associated
with $K$ that is outbound, chosen so that $N(A)$ is maximal.
We define $\support(K) = \unconf(\outlet(K), K)$,
the set of components associated with $K$ that are 
not confined to $\outlet(K)$. By Lemma~\ref{lem:pmc_outbounds},
every member of $\support(K)$ is inbound.

We call a full block $(N(C), C)$ an {\em I-block}
if $C$ is inbound and $|N(C)| \leq k$.
We call it an {\em O-block} if $C$ is outbound and 
$|N(C)| \leq k$.

We say that an I-block $(N(C), C)$ 
is {\em feasible} if $C$ is feasible.
We say that an O-block $(N(A), A)$ is feasible
if $N(A) = \bigcup_{C \in \CC} N(C)$
for some set $\CC$ of feasible inbound components.
Note that this definition of feasibility of an O-block 
is somewhat weak in the sense that we do not require
every inbound component associated with $N(A)$ to be feasible.

We say that a potential maximal clique $\Omega$ is 
{\em feasible} if 
$|\Omega| \leq k + 1$ and every $C \in \support(\Omega)$ is feasible.  

In order to formulate mutual recurrences among
feasible I-blocks, O-blocks, and potential maximal cliques,
we need the following auxiliary notion of {\em buildable}
potential maximal cliques.

Let $\Omega$ be a potential maximal clique
with $|\Omega| \leq k + 1$.
For each $C \in \support(\Omega)$, block
$(N(C), C)$ is an I-block, since $C$ is inbound as observed above and 
we have $|N(C)| \leq k$ by our assumption that $|\Omega| \leq k + 1$.
We say that $\Omega$ is {\em buildable}
if $|\Omega| \leq k + 1$ and either 
\begin{enumerate}
\item $\Omega = N[v]$ for some $v \in V(G)$, 
\item there is some subset $\CC$ of $\support(\Omega)$ such that
$\Omega = \bigcup_{D \in \CC} N(D)$ and every member of
$\CC$ is feasible, or
\item $\Omega = N(A) \cup (N(v) \cap A)$ for some feasible O-block
$(N(A), A)$ and a vertex $v \in N(A)$.
\end{enumerate}
It will turn out that 
every feasible potential maximal clique is buildable
(Lemma~\ref{lem:buildable_feasible}). 

\begin{lemma}
\label{lem:PMC_feasible}
We have $\tw(G) \leq k$ if and
only if $G$ has a feasible potential maximal clique
$\Omega$ with $\outlet(\Omega) = \emptyset$.
\end{lemma}
\begin{proof}
Suppose first that $G$ has a feasible potential 
maximal clique
$\Omega$ with $\outlet(\Omega) = \emptyset$.
Note that $\support(\Omega) = \CC(\Omega)$, as every
$C \in \CC(\Omega)$ is unconfined to an empty set. 
For each component $C \in \support(\Omega)$,
let $T_C$ be the tree-decomposition of $G\langle C \rangle$ of width
$k$ or smaller,
which exists since $C$ is feasible by the
definition of a feasible potential maximal clique.
Let $X_C$ be a bag of $T_C$ such that $N(C) \subseteq X_C$. 
Combine these tree-decompositions
into a tree $T$ by adding bag $\Omega$ and
letting each $X_C$ in $T_C$ be adjacent to $\Omega$.
That $T$ satisfies the first two conditions for
tree decomposition is trivial. The third condition is
also satisfied, since, if a vertex $v$ appears
in $N[C]$ for two or more members $C$ in $\support(\Omega)$,
then $v$ appears in $X_C$ for each such $C$ and in 
$\Omega$. Therefore, $T$ is a tree decomposition of
$G$ of width $k$ or smaller and hence $\tw(G) \leq k$.

For the converse, suppose the treewidth of $G$ is $k$
or smaller.  Let $T$ be a canonical tree-decomposition
of $G$ of width $k$ or smaller: each bag of
$T$ is a potential maximal clique and the intersection
of each pair of adjacent bags of $T$ is a minimal separator. 
Orient each
edge of $T$ as follows. Let $X$ and $Y$ be adjacent
bags in $T$ and let $S = X \cap Y$.
Let $C$ be the outbound full component associated 
with the minimal separator $S$.
Then, $C$ intersects 
exactly one of $X$ and $Y$.  If $C$ intersects
$X$ then we orient the edge between $X$ and $Y$ from
$Y$ to $X$; otherwise from $X$ to $Y$. Since $T$
is a tree, the resulting directed tree has a sink
$X_0$. Then, each component $C$ associated with $X_0$ is inbound and
hence $\outlet(X_0) = \emptyset$.
We show that each such $C$ is moreover feasible.
Indeed, the required tree-decomposition of $G\langle C \rangle$ 
may be obtained from 
$T$ by taking intersection of every bag with $N[C]$:
the resulting tree is a tree-decomposition of $G[N(C)]$ and
contains the bag $X_0 \cap N[C] \supseteq N(C)$.  
The width of the tree-decomposition is not greater than that of $T$
and hence is $k$ or smaller.  Therefore, I-block $(N(C), C)$ for
each component $C$ associated with $X_0$ is 
feasible and hence the potential maximal clique $X_0$ 
is feasible.
\qed 
\end{proof}

\begin{lemma}
\label{lem:PMC_original}
Let $C$ be a connected set of $G$ 
such that $N(C)$ is a minimal separator.
Let $\Omega$ be a potential maximal clique
of $G\langle C \rangle$. Then, $\Omega$ is a potential maximal
clique of $G$. 
\end{lemma}
\begin{proof}
For each component $D$ associated with $N(C)$, let 
$H_D$ be a minimal chordal completion of $G\langle C \rangle$.
In particular, choose $H_C$ so that $\Omega$ is a clique in
$H_C$. Let $H$ be the union of these graphs: 
$V(H) = V(G)$ and 
$E(H) = \bigcup_{D \in \CC(N(C))} E(H_D)$.  It is clear that
$H$ is chordal.  Let $H'$ be a minimal chordal completion
of $G$ contained in $H$. It is well-known that every
minimal separator is a clique in every chordal completion and
hence $N(C)$ is a clique in $H'$. Therefore, the minimality of
$H_D$ for each $D$ implies that $H' = H$. As $\Omega$ is a clique
in $H_C$, it is a clique in $H$ and hence is a potential
maximal clique of $G$.
\qed 
\end{proof}

The following is our oriented version of the recurrence
in the BT algorithm described in the beginning of this section. 
\begin{lemma}
\label{lem:I-block-feasible}
An I-block $(N(C), C)$
is feasible if and only if
there is some feasible potential maximal clique $\Omega$ with
$\outlet(\Omega) = N(C)$ and 
$\bigcup_{D \in \support(\Omega)} D = C$.
\end{lemma} 
\begin{proof}
Suppose first that there is a feasible potential maximal clique 
$\Omega$ as in the lemma. For each component $D \in \support(\Omega)$,
let $T_D$ be a tree-decomposition of $G\langle D \rangle$ 
of width $k$ or smaller and $X_D$ be a bag in $T_D$
containing $N(D)$.  Combine
these tree-decompositions $T_D$, $D \in \support(\Omega)$,
into a tree $T$ by adding bag $\Omega$ and let it be adjacent to
$X_D$ for each $D \in \support(\Omega)$. 
We confirm that $T$ is a tree-decomposition of
$G[N[C]]$. Every vertex $v \in N[C]$ appears in some bag of $T$ since
$C$ is the union of $D$ for all $D \in \support(\Omega)$ and
the bag $\Omega$ contains $N(C)$. Every edge of $G[N[C]]$ appears
in some bag of $T$ for the same reason. 
The third condition for  
$T$ being a tree decomposition is also satisfied, 
since, if a vertex $v$ appears
in $N[D]$ for two or more members $D$ in $\support(\Omega)$,
then $v$ appears in $X_D$ for each such $D$ and in 
$\Omega$. Therefore, $T$ is a tree decomposition of
$G[N[C]]$ of width $k$ or smaller and hence 
the bag $\Omega$ in $T$ contains $N(C)$, 
$T$ attests the feasibility of the I-block $(N(C), C)$. 

For the converse, suppose that I-block $(N(C), C)$ is feasible.
Let $T$ be a canonical tree-decomposition of $G \langle C \rangle$ of width $k$
or smaller. 
Orient the edges of $T$ as in the proof of
Lemma~\ref{lem:PMC_feasible}: orient the edge
from $X$ to $Y$ if and only if $Y$ intersects
the outbound full component associated with $X \cap Y$.
We need to stress here that the notion of outbound
components used in this orientation is with respect
to the entire graph $G$ and not with
respect to $G\langle C \rangle$, the graph of which $T$ is a
tree-decomposition. 
As $N(C)$ is a clique in $G \langle C \rangle$, $T$ contains
a bag that contains $N(C)$. In the subtree of $T$
induced by those bags containing $N(C)$, let $X_0$ be a sink with respect to
the above orientation.  
As $T$ is canonical, $X_0$ is a potential maximal
clique of $G \langle C \rangle$ and hence of $G$ by
Lemma~\ref{lem:PMC_original}. We show below that
$X_0$ is feasible.

Let $A$ be the outbound full component associated
with $N(C)$. As $N(C) \subseteq X_0$ and $A \cap N[C] = \emptyset$,
$A$ is a component associated with $X_0$.
We claim that $N(C) = \outlet(X_0)$.  Suppose otherwise
that there is some outbound component $A'$ associated
with $X_0$ such that $N(C)$ is a proper subset of
$N(A')$.  Then, as $A'$ is not confined to $N(C)$,
$C = \crib(N(C), X_0)$ contains $A'$.
Therefore, there is some bag $X$ adjacent to
$X_0$ in $T$ such that $X \cap A' \neq \emptyset$.
Since $N(C)$ is a minimal separator that
separates $A$ from $A'$, $X$ must contain $N(C)$.
But, since $A'$ is an outbound component associated with
$X_0$, the edge between
$X_0$ and $X$ is oriented from $X_0$ to $X$.  This contradicts the choice
of $X_0$ and we conclude that $N(C) = \outlet(X_0)$.

It remains to verify that each 
$D \in \support(X_0)$ is feasible.
This is true since the tree of bags obtained from $T$
by intersecting each bag with $N[D]$ is a tree-decomposition
of $G \langle D \rangle$ required for the feasibility
of $D$.
\qed 
\end{proof}

\begin{lemma}
\label{lem:support-outbound}
Let $K$ be a cliquish vertex set, 
$\CC$ a non-empty subset of $\support(K)$, and 
$S = \bigcup_{C \in \CC} N(C)$.
If $S$ is a proper subset of $K$ then 
$\crib(S, K)$ is outbound.
\end{lemma}
\begin{proof}
Let $K$, $\CC$ and $S$ be as in the lemma.
Since $K$ is cliquish, $\crib(S, K)$ is a full component
associated with $S$ that contains $K \setminus S$, 
by Lemma~\ref{lem:crib}. 
To show that it is outbound, it suffices to show that
no other full component associated with $S$ is 
outbound. Let $A$ be an arbitrary 
full component associated with $S$ that is
distinct from $\crib(S, K)$. 
As $A$ does not intersect $K$, 
it is a component associated with $K$.
Let $C$ be an arbitrary member of $\CC$.
Then, $C$ is confined to $S$ by the
definition of $S$.  On the other hand
$C$ is not confined to $\outlet(K)$ since
$C \in \support(K)$.  Therefore,
$S$ is not a subset of $\outlet(K)$.
$A$ cannot be outbound, since it would
imply that $S = N(A) \subseteq \outlet(K)$.
Therefore, $A$ is inbound and, since this holds
for every full component associated with $S$
other than $\crib(S, K)$, $\crib(S, K)$ is outbound.
\qed 
\end{proof}

The following lemma is crucial for our PID result:
the algorithm described in the next section
generates all buildable potential
maximal cliques and we need to guarantee
all feasible maximal cliques to be among them. 

\begin{lemma}
\label{lem:buildable_feasible}
Let $\Omega$ be a feasible potential maximal clique. 
Then, $\Omega$ is buildable.
\end{lemma}
\begin{proof}
Let $S = \bigcup_{C \in \support(\Omega)} N(C)$.

Suppose first that $S \cup \outlet(\Omega) \neq \Omega$ 
and let $v$ be
an arbitrary member of $\Omega \setminus (S \cup \outlet(\Omega))$. 
Since $\Omega$ is cliquish and $v$ is not in $N(C)$
for any component $C$ associated with $\Omega$, 
$v$ is adjacent to every other vertex in $\Omega$.
Therefore, $\Omega \subseteq N[v]$.
Let $C$ be an arbitrary component associated with $\Omega$.
If $C$ is confined to $\outlet(\Omega)$ then $v \not\in N(C)$
since $v \not\in \outlet(\Omega)$. Otherwise, $C \in \support(\Omega)$
and hence $v \not\in N(C)$ as $v \not\in S$.
Therefore, $N(v) \setminus \Omega$ is empty and hence
we have $\Omega = N[v]$.  Thus, $\Omega$ is buildable, the
first case of buildability.

Suppose next that $S \cup \outlet(\Omega) = \Omega$. 
We have two cases to 
consider: $S = \Omega$ and $S \neq \Omega$.

Consider the case where $S = \Omega$.
Let $\CC_0$ be an arbitrary minimal subset of $\support(\Omega)$ such that
$\bigcup_{C \in \CC_0} N(C) = \Omega$.  
Since $\Omega$ does not have
a full component associated with it, $\CC_0$ 
has at least two members.  Let $C_0$ be an 
arbitrary member of $\CC_0$ and let 
$\CC_1 = \CC_0 \setminus \{C_0\}$. 
From the minimality
of $\CC_0$, $S_1 = \bigcup_{C \in \CC_1} N(C)$
is a proper subset of $\Omega$.
By Lemmas~\ref{lem:crib} and \ref{lem:support-outbound}, 
$A_1 = \crib(S_1, \Omega)$ is a full component associated with $S_1$ and
is outbound. 
Therefore, $(S_1, A_1)$ is an O-block and is feasible since
every member of $\CC_1 \subseteq \support(\Omega)$ is feasible
as potential maximal clique $\Omega$ is feasible.
Thus, the second case in the definition of feasible potential maximal
cliques applies.

Finally, suppose that $S \neq \Omega$. 
Let $A = \crib(S, \Omega)$.  Then, $A$ is a full component
associated with $S$ and is outbound, by
Lemmas~\ref{lem:crib} and \ref{lem:support-outbound}.
Since $S = \bigcup_{C \in \support(\Omega)} N(C)$ and
$\Omega$ is feasible, the O-block $(S, A)$ is feasible.  
Let $x$ be an arbitrary vertex in $\Omega \setminus S$.
Since we are assuming that $S \cup \outlet(\Omega) = \Omega$ 
we have $x \in \outlet(\Omega) \setminus S$.
Let $v$ be an arbitrary vertex in $\Omega \setminus \outlet(\Omega)$.
Observe that there is no component $C$ associated with $\Omega$ such
that $N(C)$ contains both $x$ and $v$: $x \not\in N(C)$ for every $C \in \support(\Omega)$ and $v \not\in N(C)$
for every $C$ that is confined to $\outlet(\Omega)$.
Since $\Omega$ is cliquish, it follows that
$x$ and $v$ are adjacent to each other. 
Therefore, we have $\Omega \setminus S
\subseteq N(v)$.
Moreover, $A$ contains $\Omega \setminus S$ by Lemma~\ref{lem:crib}.
Finally, $A \setminus \Omega$ is disjoint from 
$N(v)$, since every component $D$ associated with $\Omega$ 
such that $v \in N(D)$ is not confined to $\outlet(\Omega)$ and
hence contained in $C$. 
Therefore, we have $\Omega = S \cup (N(v) \cap A)$, and the
third case in the definition of buildable potential maximal cliques
applies.
\qed 
\end{proof}

\section{Algorithm} \label{sec:algorithm}
Given graph $G$ and positive integer $k$,
our algorithm generates all I-blocks,
O-blocks, and potential maximal cliques that are feasible.
In the algorithm description below, the
following variables, with suffixes, are used:
$\II$ for listing feasible
I-blocks, $\OO$ for feasible O-blocks,
$\PP$ for buildable potential maximal cliques, and
$\Ss$ for feasible potential maximal cliques.
We note that each member of $\II$ and $\OO$ is actually
the component part of an I- or O-block.

\bigskip
\noindent\textbf{Algorithm PID-BT}

\smallskip
\noindent\textbf{Input:}
Graph $G$ and positive integer $k$
\newline\noindent\textbf{Output:}
``YES'' if $\tw(G) \leq k$; ``NO'' otherwise
\newline\noindent\textbf{Procedure:}
\begin{enumerate}
  \item Let $\II_0 = \emptyset$ and $\OO_0 = \emptyset$.
  \item Initialize $\PP_0$ and $\Ss_0$ to $\emptyset$.
  \item Set $j = 0$.
  \item For each $v \in V(G)$, if $N[v]$ is a potential maximal clique
  with $|N[v]| \leq k + 1$ then add $N[v]$ to $\PP_0$ and 
  if, moreover, $\support(N[v]) = \emptyset$ then do the following.
  \begin{enumerate}
    \item Add $N[v]$ to $\Ss_0$. 
    \item If $\outlet(N[v]) \neq \emptyset$ then let $C = \crib(\outlet(N[v]),
    N[v])$ and, provided that $C \neq C_h$ for $1 \leq h \leq j$, 
    increment $j$ and let $C_j = C$.
  \end{enumerate}
  \item Set $i = 0$.
  \item Repeat the following and stop repetition when $j$ is not incremented
  during the iteration step.
  \begin{enumerate}
    \item While $i < j$, do the following.
    \begin{enumerate}
      \item Increment $i$ and let $\II_i$ be $\II_{i - 1} \cup
      \{C_i\}$.
      \item Initialize $\OO_i$ to $\OO_{i - 1}$, $\PP_i$ to $\PP_{i - 1}$, and
      $\Ss_i$ to $\Ss_{i - 1}$.  
      \item For each $B \in \OO_{i - 1}$ such that
      $C_i \subseteq B$ and $|N(C_i) \cup N(B)| \leq k + 1$, 
      let $K = N(C_i) \cup N(B)$ and do the following.
      \begin{enumerate}
        \item If $K$ is a potential maximal clique, 
      	    then add $K$ to $\PP_i$. 
      	\item If $|K| \leq k$ and there is
             a full component $A$ associated with $K$ (which is unique), 
    	     then add $A$ to $\OO_i$.
      \end{enumerate}
      \item Let $A$ be the full component associated with $N(C_i)$ and
      add $A$ to $\OO_i$.
      \item For each $A \in \OO_i \setminus \OO_{i - 1}$ and $v \in N(A)$, let 
      $K = N(A) \cup (n(v) \cap A)$ and if $|K| \leq k + 1$ and
      $K$ is a potential maximal clique then add $K$ to $\PP_i$.
      \item For each $K \in \PP_i \setminus \Ss_{i - 1}$, 
       if $\support(K) \subseteq \II_i$ then add $K$ to $\Ss_i$ and do the
       following:
       if $\outlet(K) \neq \emptyset$ then let 
    	    $C = \crib(\outlet(K), K)$ and, provided that $C \neq C_h$ for $1 \leq
    	    h \leq j$, increment $j$ and let $C_j = C$.
    \end{enumerate}
  \end{enumerate}
  \item If there is some $K \in \Ss_j$ such that
  $\outlet(K) = \emptyset$, then answer ``YES'';
  otherwise, answer ``NO''.
\end{enumerate}
\begin{theorem}
\label{thm:correctness}
Algorithm PID-BT, given $G$ and $k$, answers
``YES'' if and only if $\tw(G) \leq k$.
\end{theorem}
\begin{proof}
We show that $\Ss_J$ computed by the algorithm, 
where $J$ denotes the final value of $j$, is exactly the
set of feasible potential maximal cliques for the given $G$ and $k$. The theorem then follows by
Lemma~\ref{lem:PMC_feasible}.

In the following proof, $\OO_i$, $\PP_i$, and
$\Ss_i$ for each $i$ stand for the final values of these program variables.

We first show by induction on $i$ that the following conditions are satisfied.
\begin{enumerate}
  \item For every $1 \leq h \leq i$, $(N(C_j), C_j)$ is a feasible I-block.
  \item $\II_i = \{C_h \mid 1 \leq h \leq i\}$.
  \item For every $A \in \OO_i$, $(N(A), A)$ is a 
  feasible O-block.
  \item Every $K \in \PP_i$ is a buildable potential maximal clique.
  \item Every $K \in \Ss_i$ is a feasible potential maximal clique.  
\end{enumerate}

Consider the base case $i = 0$.
Condition 1 vacantly holds. Conditions 2 and 3 also hold since
$\II_0 = \OO_0 = \emptyset$. Condition 4 holds:
$N[v]$ is confirmed to be a potential maximal clique before it is
added to $\PP_0$ and is buildable by the definition of buildability
(case 1).
Condition 5 holds since $\support(N[v]) = \emptyset$ implies that the
potential maximal clique $N[v]$ is feasible.

Suppose $i > 0$ and that the above conditions are satisfied for smaller values
of $i$.
\begin{enumerate}
  \item When $C_i$ is defined, there is some $i' < i$ and $K \in
  \Ss_{i'}$ such that $\outlet(K) \neq \emptyset$ and 
  $C_i = \crib(\outlet(K), K)$. By the induction hypothesis, 
  $K$ is a feasible potential maximal clique and hence, by
  Lemma~\ref{lem:I-block-feasible}, $(N(C_i), C_i)$ is a feasible I-block.
  \item As $\II_{i - 1} = \{C_h \mid 1 \leq h \leq i - 1\}$ and 
  $\II_{i} = \II_{i - 1} \cup \{C_i\}$,  $\II_i = \{C_h \mid 1 \leq h \leq i\}$
  holds.
  \item Let $A \in \OO_i \setminus \OO_{i - 1}$.
  Then there is some $B \in \OO_{i - 1}$ such that $A$ is outbound, $|N(A)| \leq
  k$, and $N(A) = N(C_i) \cup N(B)$. From the first two conditions, 
  $(N(A), A)$ is an O-block. By the
  induction hypothesis, $(N(B), B)$ is a feasible O-block and hence  
  $N(B) = \bigcup_{D \in \CC} N(D)$ for some set
  $\CC$ of feasible inbound components. As $C_i$ is feasible
  by 1 above and $N(A) = \bigcup_{D \in \CC \cup \{C_i\}} N(D)$,
  O-block $(N(A), A)$ is feasible. 
  \item Let $K \in \PP_i \setminus \PP_{i - 1}$.
  Then, $K$ is added to $\PP_i$ either at step 
  6-(a)-iii-A or at step 6-(a)-v. Consider the first case,
  Then, $K = N(B) \cup N(C_i)$ where $(N(B), B)$ is a feasible O-block
  and hence $N(B) = \bigcup_{D \in \CC} N(D)$ for some set
  $\CC$ of feasible inbound components. As $C_i$ is feasible,
  $K$ satisfies all the conditions in the second case of the definition of
  buildable potential maximal cliques. Consider next the second case,
  $K$ is obtained at step 6-(a)-v. Then, $K = N(A) \cup (n(v) \cap A)$, where
  $(N(A), A)$ is a feasible O-block, and the third case in the definition of
  buildable potential maximal cliques applies. 
  \item Let $K \in \Ss_i \setminus \Ss_{i-1}$. Then, $K \in \PP_i$ and
  is a buildable potential maximal clique by 4 above. The confirmed
  condition $\support(K) \subseteq \II_i$ ensures that $K$ is feasible,
  since every member of $\II_i$ is feasible by 1 and 2 above.
\end{enumerate}
We conclude that every member of $\Ss_J$ is a feasible
potential maximal clique.

In showing the converse, the following observation is crucial.
Let $(N(A), A)$ be a feasible O-block such that
$N(A) = \bigcup_{C \in \CC} N(C)$ for some set $\CC$ of
feasible components and suppose $\CC \subseteq \II_i$.
Then, $A \in \OO_i$. The proof is a straightforward induction on $i$.

The proof of the converse consists in showing the following by induction on $m$.
\begin{enumerate}
  \item For each feasible I-block $(N(C), C)$, with $|C| = m$, 
  there is some $i$ such that $C = C_i$.
  \item For each feasible O-block $(N(A), A)$ with $|A| = |V(G)| - m$, 
  there is some $i$ such that $A \in \OO_i$.
  \item For each buildable potential maximal clique $\Omega$
  such that $|\bigcup_{C \in \support(\Omega)} C| = m$,
  there is some $i$ such that $\Omega \in \PP_i$. 
  \item For each feasible potential maximal clique $\Omega$ 
  such that $|\bigcup_{C \in \support(\Omega)} C| = m$,
  there is some $i$ such that $\Omega \in \Ss_i$.
\end{enumerate} 
The base case $m = 0$ is vacantly true. Suppose
$m > 0$ and the statements hold for smaller values of $m$.
\begin{enumerate}
  \item 
Let $(N(C), C)$ be a feasible I-block with $|C| = m$.
Then, by Lemma~\ref{lem:I-block-feasible}, there
is some feasible potential maximal clique $\Omega$
such that $N(C) = \outlet(\Omega)$ and $C = \crib(N(C), \Omega)$.
We have $|\bigcup_{C \in \support(\Omega)} C| < m$,
since this union is a subset of $C \setminus (\Omega \setminus N(C))$.
Therefore, by the induction hypothesis, there is some $i$ such
that $\Omega \in \Ss_i$.  Therefore, 
$C$ is constructed as $C_j$ either at step 4-(b) or at step 6-(a)-vi. 
\item Let $(N(A), A)$ be a feasible O-block with
$|A| = |V(G)| - m$. Let $\CC$ be a set of feasible components such that
$N(A) = \bigcup_{C \in \CC} N(C)$ and let $C$ be an arbitrary member
of $\CC$. As $C$, $A$, and $N(C)$ are pairwise disjoint, we have
$|C| < m$.  Therefore, there is some $i_C$ such that
$C_{i_C} = C$. Set $i = \max\{i_C \mid C \in \CC\}$.
Then, $\CC \subseteq \II_i$ and hence $A \in \OO_i$, by the observation above.
\item Let $\Omega$ be a buildable potential maximal clique with 
$|\bigcup_{C \in \support(\Omega)} C| = m$.
In the first case of the definition of buildability, $\Omega$ is added
to $\PP_0$ at step 4. In the second case, 
we have $\Omega = \bigcup_{C \in \CC} N(C)$ for some   
$\CC \subseteq \support(\Omega)$ such that every member of $\CC$ is feasible.
Choose $\CC$ to be minimal subject to these conditions.
Let $C$ be an arbitrary member of $\CC$.
As $|C| \leq m$, by the induction hypothesis and 1 above, there is
some $i_C$ such that $\CC \subseteq \II_{i_C}$. 
Choose $C \in \CC$ so that $i_C$ is the largest and let the chosen
be $D$. Let $\CC' = \CC \setminus \{D\}$ and let $S = \bigcup_{C \in \CC'}
N(C)$.
By the minimality of $\CC$, $S$ is a proper subset of $\Omega$.
Therefore, $\crib(S, \Omega)$ is a full component associated with $S$ and
there is an outbound full component $A$ associated with $S$. 
As all members of 
$\CC'$ is feasible and $|S| \leq k$, $(S, A)$ is a feasible O-block. By
the choice of $D$, we have $\CC' \subseteq \II_{i_D - 1}$ and hence  
$A \in \OO_{i_D - 1}$ by the observation above. At step 6-(a)-iii-A in the
iteration for $i = i_D$, $\Omega$ is put into $\PP_{i_D}$.
\item Let $\Omega$ be a feasible potential maximal clique with 
$|\bigcup_{C \in \support(\Omega)} C| = m$.
Then, by 3 above, there is some $i_1$ such that 
$\Omega \in \PP_{i_1}$.  Furthermore, as every member $C$ of $\support(\Omega)$
is feasible and $|C| \leq m$, there is some $i_2$ such that
$\support(\Omega) \subseteq \II_{i_2}$, by 1 above.
At step 7 in the iteration for $i = \max\{i_1, i_2\}$, $\Omega$ is put into
$\Ss_i$. 
\end{enumerate}
We conclude that every feasible potential maximal clique is
in $\Ss_J$. This completes the proof.
\qed 
\end{proof}

\section{Running time analysis}
\label{sec:time}
The running time of our algorithm is stated in terms of the
the number of positive subproblem instances. Given $G$ and
$k > 0$, let $\II_G^k$ denote the set of feasible I-blocks and
$\OO_G^k$ the set of feasible O-blocks.

\begin{observation}
\label{obs:run_time}
Given $G$ and $k > 0$, algorithm PID-BT runs in
$O^*(|\II_G^k|\cdot|\OO_G^k|)$ time.
\end{observation} 
\begin{proof}
The number of iteration in step 6, where $i$ is incremented
each time, is $|\II_G^k|$. In each iteration step,
every computation step may be charged to each element
of $\OO_{i - 1}$ and the total number of steps charged to 
a single element of $\OO_{i - 1}$ is $n^{O(1)}$.
Since $|\OO_{i - 1}| \leq |\OO_G^k|$, we have the claimed time bound. 
\qed 
\end{proof}

The bound in this observation is incomparable to the 
previous bounds on non-PID versions of the BT algorithm,
which run in $O^*(|\Pi_G|)$ time when $\Pi_G$, 
the set of potential maximal cliques in $G$, is given.
In \cite{FV12},
in addition to a combinatorial bound of $|\Pi_G|= O(1.7549^n)$,
it was shown that $\Pi_G$ can be computed in $O^*(\Pi_G)$
time.

It should be emphasized, however, that it is not known whether
the decision problem version of the treewidth problem with
given $k$ can be solved in 
$O^*(|\Pi_G^{k+1}|)$ time, where 
$\Pi_G^k$ is the set of potential maximal cliques of
cardinality at most $k$ in $G$. The bottleneck here is
the time to list all members of $\Pi_G^{k+1}$.
Although a nontrivial upper bound on $|\Pi_G^{k+1}|$ in terms of
$n$ and $k$, together with a running time bound based on it, 
is given in \cite{FV12}, a huge gap between
the actual value $|\Pi_G^{k+1}|$ and the upper bound is observed in
practice, as shown in the next section.
This is the gap that makes the bound in Obseravation~\ref{obs:run_time}
interesting.

\section{Experimental analysis}
\label{sec:experimental}
To study the strength of the running time bound of
Observation~\ref{obs:run_time} from a practical view point, 
we have performed some experiments, in which 
we count the number of combinatorial objects
involved in the treewidth computation.
We first compare the actual number of relevant potential maximal
cliques (that is, of cardinality at most $k + 1$ where
$k$ is the treewidth) 
with the theoretical uppser
bounds on that number: the naive bound of $\tbinom{n}{k + 1}$ and
an assymptotically stronger bound of
$n (\tbinom{\lceil (2n + k + 7)/3 \rceil}{k
    + 2}  + 
    \tbinom{\lceil (2n + k + 4)/2 \rceil}
    {k + 1})$ given in \cite{FV12}.
Table~\ref{tab:pmcs} shows the results
on some random instances, where the number of vertices 
$n$ is 20, 30, 40, or 50, the number of edges $m$ is $2n$, $3n$, 
$4n$ or $5n$, and the graph for each pair $(n, m)$
is chosen uniformly at random
from the set of all graphs with $n$ vertices and $m$ edges.
Huge gaps between the actual number and the upper bounds 
are apparent.

\begin{table}[H]
  \begin{center} 
    \begin{tabular}{|r|r|r|r|r|r|} \hline 
    $n = |V|$& $|E|$ & $k = \tw$ & 
    PMCs ($\leq k + 1$)
    & $\tbinom{n}{k + 1}$ 
    &  $n (\tbinom{\lceil (2n + k + 7)/3 \rceil}{k
    + 2}  + 
    \tbinom{\lceil (n + k + 4)/2 \rceil}
    {k + 1})$ \\
\hline    
    20& 40   & 6 & 115& 77520& 1003860\\ \hline  
    20& 60   & 8 & 96& 167960& 2076360\\ \hline
20& 80  & 11 &  121& 125970& 1921680\\ \hline
20& 100 & 11 &  37& 125970& 1921680\\ \hline
30& 60 & 7&  559& 5852925& 67393950\\ \hline
30& 90  & 11 &  682& 86493225& 352580340\\ \hline
30& 120 & 14 &  1137& 155117520&  430361970\\ \hline 
30& 150 & 16 &  768& 119759850&   426140550\\ \hline 
40& 80   & 8 &  5341& 273438880&  2705471600\\ \hline 
40& 120 & 14 &  10372& 40225345056&  91260807600\\ \hline 
40& 160 & 18 &  17360& 131282408400&  135562547400\\ \hline 
40& 200 & 20 &  6820& 131282408400&  157012867200\\ \hline 
50& 100 & 10 & 6029& 37353738800&  201991095800\\ \hline
50& 150 & 16 & 48068& 9847379391150&  10332510412500\\ \hline 
50& 200 & 20 & 36388& 67327446062800&  53246262826500\\ \hline 
50& 250 & 24 & 47729& 126410606437752& 52230760068000 \\ \hline
\end{tabular}   \end{center}
    \caption{The numbers of relevant potetntial maximal cliques and 
    their upper bounds}
	\label{tab:pmcs}
\end{table}

Since the running time bound in Observation~\ref{obs:run_time} involves
the quantity $|\OO_G^k|$ which is not theoretically upper-bounded
by a function of $|\Pi_G^{k + 1}|$, the gaps observed in Table~\ref{tab:pmcs}
alone may not be sufficient to support the importance of this
running time bound. To address this issue,
we have counted more combinatorial objects involved in our PID computation
on the same graph instances: in addition to relevant potential
maximal cliques counted above, all potential maximal cliques, 
relevant minimal separators, all minimal separators, 
feasible I-blocks, feasible O-blocks 
and feasible potential maximal cliqeus.
Here, the input $k$ to the decision problem is set to the treewidth of the
graph.

Table~\ref{tab:numObjects} shows the result. 
We see that the number of feasible O-blocks is
smaller than the number of relevant potential mmaximal cliques,
as far as these instances are concerend. This, together with
what we have observed in Table~\ref{tab:pmcs}, provides an evidence that
the running time bound of Observation~\ref{obs:run_time} is 
more relevant from a practical point of view than 
the running time bounds of known theoretical algorithms.

We also see that the number of all potential maximal
cliques grows much faster than the number of relevant potential maximal
cliques.  This shows the advantage of our algorithm which avoids
generating all potential maximal cliques.

To summarize, our PID algorithm has advantages over the
standard BT algorithms
because the running time upper bounds of those algorithms 
are either in terms of a 
combinatorial {\em upper bound} on the number of relevant potential maximal
cliques or in terms of the actual number of {\em all} potential maximal cliques:
our experiments reveal huge gaps between the actual number of
relevant potential maximal cliques and both of these
quantities.
Note that, if there is an efficient method of generating relevant
potential maximal cliques, a non-PID version of the BT algorithm
might outperform our PID version.

\begin{table}[H]
  \begin{center} 
    \begin{tabular}{|r|r|r|r|r|r|r|r|r|r|} \hline
    & & & \multicolumn{2}{|c|}{minimal separators}&
    \multicolumn{2}{|c|}{PMCs}&\multicolumn{3}{|c|}{feasible objects} \\
    \cline{4-10}  $|V|$& $|E|$ & $\tw$ & all & $\leq\tw$ & all &
    $\leq \tw + 1$& I-blocks & O-blocks & PMCs \\ \hline 
    20& 40   & 6 &  98& 51 & 376& 115& 19 & 26 & 37 \\ \hline 
    20& 60   & 8 &  191& 48 & 796& 96& 46 & 108 & 93 \\ \hline
20& 80  & 11 &  185& 122 & 698& 376& 121& 158& 370 \\ \hline
20& 100 & 11 &  107& 25 &  354& 37& 24& 32& 36 \\ \hline
30& 60 & 7&  535& 185 & 3122& 559& 114 & 170 & 334\\ \hline
30& 90  & 11 &  2983& 247 & 20154& 682& 228& 708& 618 \\ \hline
30& 120 & 14 &  2713& 376 & 16736& 1137& 352& 804& 1055 \\ \hline 
30& 150 & 16 &  1913& 281 &  10535& 768& 240& 498& 647  \\ \hline 
40& 80   & 8 &  14842& 1070 & 178661& 5341& 840& 2965& 4154 \\ \hline 
40& 120 & 14 &  164773& 2356 & 1740644& 10372& 2080& 8637& 8577\\ \hline 
40& 160 & 18 &  134485& 3952 & 1251656& 17360& 3289& 10023& 13646 \\ \hline 
40& 200 & 20 &  52182& 1790 & 423691& 6820& 1502& 4749& 5347 \\\hline 
50& 100 & 10 & 96499 &1361 &1123621 &6029&779 & 2171 & 2914 \\ \hline
50& 150 & 16 & 1792713 & 9152 &$>$2000000 & 48068& 8099 & 36881& 39803 \\ \hline
50& 200 & 20 & 2130811& 7878 &$>$2000000 & 36388& 6956 & 28247& 29842 \\ \hline
50& 250 & 24 & 1452449 & 10571&$>$2000000 & 47729& 8949 & 30834& 37115 \\ \hline
\end{tabular}   \end{center}
    \caption{The numbers of principal objects in treewidth computation}
	\label{tab:numObjects}
\end{table}

\section{Implementation}
\label{sec:implementation}
In this section, we sketch two important ingredients of our
implementation.  Although both are crucial in obtaining
the result reported in Section~\ref{sec:performance}, 
our work on this part is preliminary and improvements
are the subject of future research.

\subsection{Data structures}
The crucial elementary operation in our algorithm is the following.
We have a set $\OO$ of feasible O-blocks obtained so far and,
given a new feasible I-block $(N(C), C)$, need to 
find all members $(N(A), A)$
of $\OO$ such that $C \subseteq A$ and $|N(C) \cup N(A)| \leq k + 1$.
As the experimental analysis in the previous section shows,
there is only a few such $A$ on average for the tested instances
even though $\OO$ is usually huge. To support an efficient
query processing, we introduce an abstract data structure
we call a block sieve.

Let $G$ be a graph and $k$ a positive
integer.  A {\em block sieve} for graph $G$ and width $k$ is a data structure
storing vertex sets of $V(G)$ which supports the following operations.
\begin{description}
\item[store($U$)]: store vertex set $U$ in in the block sieve.
\item[supersets($U$)]: return the list of entries $W$ stored in the block sieve
such that $U \subseteq W$ and $|N(U) \cup N(W)| \leq k + 1$.
\end{description}
Data structures for superset query
have been studied \cite{Savnik13}. The second condition above on 
the retrieved sets, however, appears to make this data structure new.
For each $U \subseteq V(G)$, we define the {\em margin} of $U$
to be $k + 1 - |N(U)|$. Our implementation of block sieves described 
below exploits an upper bound on the margins of vertex sets stored in the
sieve.

We first describe how such block sieves with upper bounds on margins
are used in our algorithm. 
Let $\OO$ be the current set of O-blocks. 
We use $t$ block sieves $\BB_1$, \ldots, $\BB_t$, each $\BB_i$
having a predetermined upper bound $m_i$ on the margins of
the sets stored. We have $0 < m_1 < m_2 < \ldots < m_t = k$.
We set $m_0 = 0$ for notational ease below.
In our implementation, we 
choose roughly $t = \log_2 k $ and $m_i = 2^i$ for $0 < i < t$.
For each $(N(A), A)$ in $\OO$, $A$ is stored in 
$\BB_i$ such that the margin $k + 1 - |N(A)|$ is $m_i$
or smaller but larger than $m_{i - 1}$.
When we are given an I-block $(N(C), C)$ and are to
list relevant blocks in $\OO$, we query all of the
$t$ blocks with the operations $\mathop{\rm supersets}(C)$.
These queries as a whole return the list of all vertex sets $A$ 
such that $(N(A), A) \in \OO$, $C \subseteq A$, and 
$|(N(A)\cup N(C))| \leq k + 1$. 

We implement a block sieve by a trie $\TT$. The upper bound $m$
on margin is not used in the construction of the sieve; it is
used in the query time.  In the following, we assume
$V(G) = \{1, \ldots, n\}$ and, by an interval
$[i, j]$, $1 \leq i \leq j \leq n$, we mean the
set $\{v: i \leq v \leq j\}$ of vertices.
Each non-leaf node $p$ of $\TT$ is labelled with
a non-empty interval $[s_p, f_p]$, such that
$s_r = 0$ for the root $r$, $s_p = f_q + 1$ if $p$ is a child of
$q$, and $f_p = n$ if $p$ is a parent of a leaf. 
Each edge $(p, q)$ which connects
node $p$ and a child $q$ of $p$, is labelled with a subset $S_{(p, q)}$
of the interval $[s_p, f_p]$.  Thus, for each node $p$,
the union of the labels of the edges along the path from the root
to $p$ is a subset of the interval $[1, s_p - 1]$, 
or $[1, n]$ when $p$ is a leaf, which we denote
by $S_p$.  
The choice of interval $[s_p, f_p]$ for each node $p$ is heuristic.
It is chosen so that the number of descendants of $p$ is not too large
or too small.  In our implementation, the interval size is
adaptively chosen from $8$, $16$, $32$, and $64$.

Each leaf $q$ of trie $\TT$ represents a single set stored at this 
leaf,
namely $S_q$ as defined above. We denote by $S(\TT)$ the set of all sets
stored in $\TT$.  Then, for each node
$p$ of $\TT$, the set of sets stored under $p$ is 
$\{U \mid U \cap [1, p] = S_p\}$.

We now describe how a query is processed against this 
data structure.  Suppose query $U$ is given.
The goal is to visit all leaves $q$ such that 
$U \subseteq S_q$ and $|N(U) \cup N(S_q)| \leq k + 1$.
This is done by a depth-first traversal of the trie $\TT$.
When we visit node $p$, we have the invariant
that $U \cap [1, f_p] \subseteq S_p$, since otherwise
no leaf in the subtree rooted at $p$ stores a superset of $U$.
Therefore, we descend from $p$ to a child $p'$ of $p$ only
if this invariant is maintained.
Moreover, we keep track of the quantity $i(p, U) =
|N(U) \cap S_p|$ in order to make further pruning of 
search possible.  For each leaf $q$
below $p$ such that $U \subseteq S_q$, 
we have $i(q, U) \geq i(p, U)$.
Combining this with eauality $|N(U) \setminus N(S_q)| =
|N(U) \cap S_q| = i(q, U)$, we have
$|N(U) \cup N(S_q)| \geq |N(S_q)| + i(p, U)$.
Since we know an upper bound $m$ on the margin 
$k + 1 - |N(S_q)|$ of $S_q$, or lower bound $k + 1 - m$ on
$|N(S_q)|$,
we may prune the search under node $p$ if $i(p, U) > m$,
since this inequality implies $|N(U) \cup N(S_q)| > k + 1$
for every leaf $q$ under $p$.
When we reach a leaf $q$, we test if $|N(U) \cup N(S_q)| \leq k + 1$ indeed
holds.

\subsection{Safe separators}
The notion of safe separators for treewidth was introduced by Bodlaender and
Koster \cite{BK06}: a separator $S$ of $G$ is {\em safe} if completing $S$
into a clique does not change the treewidth of $G$. If we find a safe separator
$S$ then the problem of deciding tree width of $G$ reduces to that of deciding
the treewidth of $G\langle C \rangle$ for each component $C$ associated with
$S$. Preprocessing $G$ into such independent subproblems is highly desirable
whenever possible. 

The above authors observed that  
a powerful sufficient condition for safeness can be formulated based on graph
minors.  A {\em labelled minor} of $G$ is a graph obtained from $G$ by zero
or more applications of the following operations. (1) Edge contraction: choose 
an edge $\{u, v\}$, replace $u$ and $v$ by a single new vertex and let all
neighbors of $u$ and $v$ be adjacent to this new vertex; name the new vertex
as either $u$ or $v$.  (2) Vertex deletion: delete a vertex together with all
incident edges. (3) Edge deletion.
\begin{lemma}
\label{lem:minoar-safe}
(\cite{BK06}) A separator $S$ of $G$ is safe if, for every component $C$
associated with $S$, $G[V(G) \setminus C]$ contains clique $S$ as a labelled
minor.
\end{lemma}

Call a separator {\em minor-safe} if it satisfies the sufficient condition for
safeness stated in this lemma. Bodlaender and Koster \cite{BK06} showed that if 
$S$ is a minimal separator and is an almost clique (deleting some single vertex
makes it a clique) then $S$ is minor-safe and moreover that the set of all
almost clique minimal separators can be found in $O(n^2 m)$ time, where
$n$ is the number of vertices and $m$ is the number of edges.

We aim at capturing as many minor-safe separators as possible, at the expense of
theoretical running time bounds on the algorithm for finding them.  Thus, in our
approach, both the algorithm for generating candidate separators and
the algorithm for deciding minor-safeness are heuristic.  For candidate
generation, we use greedy heuristic for treewidth such as min-fill and
min-degree: the separators in the resulting tree-decomposition
are all candidates for safe separators.

When we apply our heuristic decision algorithm for minor-safeness to candidate
separator $S$, one of the following occurs.
\begin{enumerate}
  \item The algorithm answers ``YES''. In this case, a required
  labelled clique minor has been found for every component associated $S$ and
  hence $S$ is minor-safe.
  \item The algorithm answers ``DON'T KNOW''. In this case, the
  algorithm has failed to find a labelled clique minor for at least one
  component, and hence it is not known if $S$ is minor-safe or not.
  \item The algorithm aborts, after reaching the prescribed number of execution
  steps.
\end{enumerate}

Our heuristic decision algorithm works in two phases.
Let $S$ be a separator, $C$ a component associated with $S$, and $R = V(G)
\setminus (S \cup C)$.
In the first phase, we contract edges in $R$ and obtain a graph $B$ on vertex
set $S \cup R'$, where each vertex of $R'$ is a contraction of some
vertex set of $R$ and $B$ has no edge between vertices in $R'$.
For each pair $u, v$ of distinct vertices in $S$, let $N(u, v)$ denote the common
neighbors of $u$ and $v$ in graph $B$.
The contractions are performed with the goal of making $|N(u, v) \cap R'|$ large
for each missing edge $\{u, v\}$ in $S$. In the second phase, for each
missing edge $\{u, v\}$, we choose a common neighbor $w \in N(u, v) \cap R'$
and contract either $\{u, w\}$ or $\{v, w\}$.  The choice of the next missing
edge to be processed and the choice of the common neighbor are done as follows.
Suppose the contractions in the second phase are done for some missing edges
in $S$.  For each missing edge $\{u, v\}$ not yet ``processed'', let
$N'(u, v)$ be the set of common neighbors of $u$ and $v$ that
are not yet contracted with any vertex in $S$.  We choose $\{u, v\}$ with
the smallest $|N'(u, v) \cap R'|$ to be processed next.  Tie-breaking when
necessary and the choice of the common neighbor $w$ in $N'(u, v) \cap R'$ to be
contracted with $u$ or $v$ is done in such a way that the minimum of $|(N'(x, y)
\cap R') \setminus \{w\}|$ is maximized over all remaining missing edges $\{x,
y\}$ in $S$.

The performance of these heuristics strongly depends on the instances.
For PACE 2017 public instances, they work quite well. Table~\ref{tab:PACE-safe}
shows the preprocessing result on the last 10 of those instances.
See Section~\ref{sec:performance} for the description of those instances and
the computational environment for the experiment. 
For each instance, the number of safe separators found and the maximum
subproblem size in terms of the number of vertices, after the graph
is decomposed by the safe separators found, are listed. 
The results show that 
these instances, which are deemed the hardest among
all the 100 public instances, are quickly decomposed into manageable subproblems
by our preprocessing.

\begin{table}[H]
  \begin{center} 
    \begin{tabular}{|c|r|r|r|r|r|r|} \hline
    name & $|V|$& $|E|$ & $tw(G)$ & safe separators found & max subproblem 
    & time(secs)
    \\
    \hline  
ex181 & 109 & 732 & 18 & 18& 89 &0.078\\ \hline  
ex183 & 265 & 471 & 11 & 173 & 76 & 0.031\\ \hline  
ex185 & 237 & 793 & 14 & 142 & 52 & 0.046\\ \hline  
ex187 & 240 & 453 & 10 & 138& 81 & 0.031\\ \hline  
ex189 & 178 & 4517 & 70 & 6& 161& 0.062\\ \hline  
ex191 & 492 & 1608 & 15 & 184& 132& 0.171\\ \hline  
ex193 & 1391 & 3012 & 10 & 791 & 119& 3.17\\ \hline  
ex195 & 216 & 382 & 10 & 114 & 84& 0.015\\ \hline  
ex197 & 303 & 1158 & 15 & 176 & 56 & 0.062\\ \hline  
ex199 & 310 & 537 & 9 & 157 & 131 & 0.046\\ \hline  
\end{tabular}   
\end{center}
    \caption{Safe separator preprocessing on PACE 2017 instances}
	\label{tab:PACE-safe}
\end{table}

On the other hand, these heuristics turned out useless for most of the DIMACS
graph coloring instances: no safe separators are found for those instances.
We suspect that this is not the limitation of the heuristics but is simply
because those instances lack minor-safe separators.  We need, however, further
study to get a firm conclusion.

\section{Performance results}
\label{sec:performance}
We have tested our implementation on two sets of instances.
The first set comes from the DIMACS graph coloring challenge \cite{JT93} and
has served as a standard benchmark suite for treewidth in the literature
\cite{GD04,BK06,Musliu08,SH09,BFKKT12,BJ14}. 
The other is the set of public instances posed by the exact treewidth track of 
PACE 2017 \cite{PACE17}.

The computing environment for the experiment is as follows.
CPU: Intel Core i7-7700K, 4.20GHz; RAM: 32GB; 
Operating system: Windows 10, 64bit; 
Programming language: Java 1.8; JVM: jre1.8.0\_121.
The maximum heap space size is 6GB by default and is 24GB where
it is stated so. The implementation is single threaded, except that
multiple threads may be invoked for garbage collection by JVM.
The time measured is the CPU time, 
which includes the garbage collection time.

To determine the treewidth of a given instance
we use our decision procedure with $k$ being incremented
one by one, starting from the obvious lower bound,
namely the minimum degree of the graph. 
Binary search is not used because
the cost of overshooting the exact treewidth can be huge.
We do not feel the need of using stronger lower bounds either, 
since the cost of executing the decision procedure for
$k$ below such lower bounds is usually quite small.

Table~\ref{tab:DIMACS} shows the results on 
DIMACS graph coloring instances.
Each row shows the name of the instance, the number of
vertices, the number of edges, the exact treewidth computed 
by our algorithm, CPU time in seconds, and the previously best
known upper and lower bounds on the treewidth.
Rows in bold face show the newly solved instances.
For all but three of them,
the previous best upper bound has turned out optimal: only the
lower bound was weaker. 
In this experiment, however, no knowledge of previous bounds are used
and our algorithm independently determines the exact treewidth.

The results on ``queen" instances illustrate
how far our algorithm has extended 
the practical limit of exact treewidth
computation. Queen7\_7 with
49 vertices is the largest instance previously solved, while
queen10\_10 with 100 vertices is now solved.
Also note that all previously solved instances are
fairly easy for our algorithm: all of them are solved within 10
seconds per instance and many of them within a second.

\begin{table}[H]
  \begin{center}
  {\small
    \begin{tabular}{|c|r|r|r|r|r|r|} \hline
    name & $|V|$& $|E|$ & $\tw$ & time(secs) & prev UB
    & prev LB\\
    \hline anna & 138 & 493 & 12 & 0.078&12 & 12\\ \hline  
david & 87 & 406 & 13 & 0.031 & 13 & 13\\ \hline 
{\bf DSJC125.5}& {\bf 125}& {\bf 3891}& {\bf 108} & 
{\bf 459}& {\bf 108}&{\bf 56}\\ \hline  
DSJC125.9 & 125 & 6961 & 119 & 0.062 & 119 & 119 \\ \hline  
{\bf DSJC250.9}& {\bf 250}& {\bf 27897}& {\bf 243}& {\bf 0.44}& 
{\bf 243}& {\bf 212}\\ \hline  
{\bf DSJC500.9}& {\bf 500}& {\bf 112437}& {\bf 492}& 
{\bf 14}& {\bf 492}& {\bf 433}
\\ \hline 
{\bf DSJR500.5}& {\bf 500}& {\bf 58862}& {\bf 246}& {\bf 546}& {\bf -}&
{\bf -}\\ \hline 
DSJR500.1c & 500 & 121275 & 485 & 2.12 & 485 & 485 \\ \hline  
fpsol2.i.1 & 496 & 11654 & 66 & 3.30 & 66 & 66\\ \hline  
fpsol2.i.2 & 451 & 8691 & 31 & 5.66 & 31 & 31\\ \hline  
fpsol2.i.3 & 425 & 8688 & 31 & 5.68 & 31 & 31\\ \hline  
{\bf games120}$^\dagger$ & {\bf 120} & {\bf 638} & {\bf 32} & 
{\bf 94738} &{\bf 32} &{\bf 24} \\ \hline
{\bf homer}$^\dagger$ &{\bf 561} &{\bf 1628} &{\bf 30}&
{\bf 2765} &{\bf 31} &{\bf 26} \\ \hline
huck & 74 & 301 & 10 & 0.012 & 10 & 10\\ \hline  
inithx.i.1 & 864 & 18707 & 56 & 8.10 & 56 & 56 \\ \hline  
inithx.i.2 & 645 & 13979 & 31 & 8.14 & 31 & 31\\ \hline  
inithx.i.3 & 621 & 13969 & 31 & 10 & 31 & 31\\ \hline  
jean & 80 & 254 & 9 & 0.031 & 9 & 9 \\ \hline  
miles250 & 128 & 387 & 9 & 0.000 & 9 & 9 \\ \hline  
miles500 & 128 & 1170 & 22 & 0.11 & 22 & 22  \\ \hline  
{\bf miles750}& {\bf 128}& {\bf 2113}& {\bf 36}&
{\bf 0.23}& {\bf 36}& {\bf 35}
\\ \hline 
miles1000 & 128 & 3216 & 49 & 0.33 & 49 & 49\\ \hline  
miles1500 & 128 & 5198 & 77 & 0.45 & 77 & 77 \\ \hline  
mulsol.i.1 & 197 & 3925 & 50 & 1.41 & 50 & 50  \\ \hline  
mulsol.i.2 & 188 & 3885 & 32 & 1.77 & 32 & 32  \\ \hline  
mulsol.i.3 & 184 & 3916 & 32 & 1.80 & 32 & 32 \\ \hline  
mulsol.i.4 & 185 & 3946 & 32 & 1.78 & 32 & 32 \\ \hline  
mulsol.i.5 & 186 & 3973 & 31 & 1.80 & 31 & 31 \\ \hline  
myciel2 & 5 & 5 & 2 & 0.000 & 2 & 2 \\ \hline  
myciel3 & 11 & 20 & 5 & 0.000 & 5 & 5  \\ \hline  
myciel4 & 23 & 71 & 10 & 0.015 & 10 & 10 \\ \hline  
myciel5 & 47 & 236 & 19 & 0.33 & 19 & 19 \\ \hline  
{\bf myciel6}& {\bf 95}& {\bf 755}& {\bf 35}& 
{\bf 419}& {\bf 35}& {\bf 29} \\ \hline 
queen5\_5 & 25 & 160 & 18 & 0.000 & 18 & 18 \\ \hline  
queen6\_6 & 36 & 290 & 25 & 0.031 & 25 & 25 \\ \hline  
queen7\_7 & 49 & 476 & 35 & 0.19 & 35 & 35 \\ \hline  
{\bf queen8\_8}& {\bf 64}& {\bf 728}& {\bf 45}& 
{\bf 4.16}& {\bf 45}& {\bf 25} \\ \hline  
{\bf queen9\_9}& {\bf 81}& {\bf 1056}& {\bf 58}
& {\bf 274}& {\bf 58}& {\bf 35} \\ \hline  
{\bf queen8\_12}& {\bf 96}& {\bf 1368}& {\bf 65} & 
{\bf 649}& {\bf -}& {\bf 39}\\ \hline  
{\bf queen10\_10}& {\bf 100}& {\bf 1470}& {\bf 72} & 
{\bf 20934}& {\bf 72}& {\bf 39}\\ \hline  
zeroin.i.1 & 211 & 4100 & 50 & 1.09 & 50& 50\\ \hline  
zeroin.i.2 & 211 & 3541 & 32 & 1.64 & 32& 32\\ \hline  
zeroin.i.3 & 206 & 3540 & 32 & 1.55 & 32& 31\\ \hline  
\end{tabular}
\\
Previous upper bounds from \cite{GD04} and \cite{Musliu08}; previous lower
bounds from \cite{GD04} and \cite{BWK06}.\\
$^\dagger$ 24GB heap space is used for these instances.
}
    \caption{Results on the DIMACS graph coloring instances}
    \label{tab:DIMACS}
    
\end{center}
\end{table}

Table~\ref{tab:DIMACS-LB} shows the lower bounds obtained by our
algorithm on unsolved DIMACS graph coloring instances.
Lower bound entries in bold face are improvements
over the previously known lower bounds.
Computation time of the previously best lower bounds
ranges from a few minutes to a week \cite{BWK06}.
Detailed comparison
of lower bound methods, which requires the
normalization of machine speeds,
is not intended here. Rather, the table is
meant to show the potential of our algorithm
as a lower bound procedure.

For many of the instances the improvements are
significant. It can also be seen from this table that
our algorithm performs rather poorly on 
relatively sparse graphs with a large
number of vertices.
 
\begin{table}[H]
  \begin{center} 
    \begin{tabular}{|c|r|r|r|r|r|r|r|} \hline
    & & &\multicolumn{3}{|c|}{lower bounds computed}&
    \multicolumn{2}{|c|}{previous bounds}\\
    \cline{4-8}
    name & $|V|$& $|E|$ & 1 sec & 1 min & 30 min & lower & upper\\
    \hline  
    DSJC125.1 & 125 & 736 &{\bf 25}&{\bf 30}&{\bf 36 }& 20 & 60 \\ \hline  
    DSJC250.1 & 250 & 3218 &{\bf 45}&{\bf 57}&{\bf 66 }& 43 & 167\\ \hline  
    DSJC250.5 & 250 & 15668 &{\bf 180}&{\bf 197}&{\bf 211}& 114 & 229\\ \hline
    DSJC500.1 & 500 & 12458 & - &{\bf 94}&{\bf 115}& 87 & 409 \\ \hline
    DSJC500.5 & 500 & 62624 & - &{\bf 360}&{\bf 388}& 231 & 479 \\ \hline
    DSJC1000.1 & 1000 & 49629 & - & 172&{\bf 189}& 183 & 896 \\ \hline
    DSJC1000.5 & 1000 & 249826& - &{\bf 724}&{\bf 742}& 469 & 977 \\ \hline
    DSJC1000.9 & 1000 & 449449 & - &{\bf 983}&{\bf 987} & 872 & 991 \\ \hline
    le450\_5a & 450& 5714& 29& 50 & 59 & 79 & 243\\ \hline
    le450\_5b & 450& 5734& -& 49& 57& -& 246\\ \hline
    le450\_5c & 450& 9803& -& 84& 100& 106& 265\\ \hline
    le450\_5d & 450& 9757& -& 94& 99& -& 265\\ \hline
    le450\_15a & 450& 8168& 24& 40& 49$^\dagger$& 94& 262\\ \hline
    le450\_15b & 450& 8169& 23& 32& 47$^\dagger$& -& 258\\ \hline
    le450\_15c & 450& 16680& -& 114& 132& 139& 350\\ \hline
    le450\_15d & 450& 16750& -& 112& 131& -& 353\\ \hline
    le450\_25a & 450& 8260& 11& 23& 25$^\dagger$& 96& 216\\ \hline
    le450\_25b & 450& 8263& 16& 26& 30$^\dagger$& -& 219\\ \hline
    le450\_25c & 450& 17343& 43& 89& 109& 144& 320\\ \hline
    le450\_25d & 450& 17425& -& 93& 112& -& 327\\ \hline
    myciel7 & 191& 2360& 22& 31& 35& 52& 66\\ \hline    
    queen11\_11 & 121 &1980 &{\bf 61}&{\bf 70} &{\bf 77}& 40 & 87 \\ \hline
    queen12\_12 & 144 & 2596 &{\bf 71}&{\bf 76}&{\bf 84}& 55 & 103 \\ \hline
    queen13\_13 & 169 & 3328 &{\bf 70} &{\bf 82}&{\bf 91} & 51& 121 \\ \hline
    queen14\_14 & 196 & 4186 &{\bf 74} &{\bf 87}&{\bf 98} & 55& 140 \\ \hline
    queen15\_15 & 225 & 5180 &{\bf 78}&{\bf 93}&{\bf 104} & 73 & 162 \\ \hline
    queen16\_16 & 256 & 6320 &{\bf 83}&{\bf 99}&{\bf 110} & 79 & 186 \\ \hline
    school1 & 385 & 19095&  73& 112& 125& 149 & 178 \\\hline
    school1\_nsh & 352 &14612 & 78& 105& 118& 132 &  152\\\hline 
\end{tabular}

Previous upper bounds from \cite{Musliu08}; previous lower
bounds from \cite{BWK06}.\\
$^\dagger$ out of memory before time out

\end{center}
    \caption{New lower bounds on the treewidth of unsolved DIMACS graph
    coloring instances}
	\label{tab:DIMACS-LB}
\end{table}

Table~\ref{tab:PACE2017} shows the results on PACE 2017
instances. The prefix ``ex" in the instance names
means that they are for the exact treewidth track.
Odd numbers mean that they are public instances disclosed
prior to the competition for testing and experimenting.
Even numbered instances, not in the list, are secret and to be used in
evaluating submissions. The time allowed to be spent for each instance
is 30 minutes. As can be seen from the table, our algorithm
solves all of the public instances with a large margin in time. 

\begin{table}[H]
  \begin{center}
  {\small
    \begin{tabular}{|c|r|r|r|r||c|r|r|r|r|} \hline
    name & $|V|$& $|E|$ & $\tw$ & time (secs) &
    name & $|V|$& $|E|$ & $\tw$ & time (secs)\\ \hline
ex001 & 262 & 648 & 10 & 1.48 &ex101 & 1038 & 291034 & 540 & 12 \\ \hline
ex003 & 92 & 2113 & 44 & 8.92 &ex103 & 237 & 419 & 10 & 3.01 \\ \hline
ex005 & 377 & 597 & 7 & 14 &ex105 & 1038 & 291037 & 540 & 12 \\ \hline
ex007 & 137 & 451 & 12 & 0.046 &ex107 & 166 & 396 & 12 & 1.44 \\ \hline
ex009 & 466 & 662 & 7 & 13 &ex109 & 1212 & 1794 & 7 & 43 \\ \hline
ex011 & 465 & 1004 & 9 & 0.50 &ex111 & 395 & 668 & 9 & 4.33 \\ \hline
ex013 & 56 & 280 & 29 & 15 &ex113 & 93 & 488 & 14 & 0.046 \\ \hline
ex015 & 177 & 669 & 15 & 0.046 &ex115 & 963 & 419877 & 908 & 18 \\ \hline
ex017 & 330 & 571 & 9 & 1.11 &ex117 & 77 & 181 & 13 & 18 \\ \hline
ex019 & 291 & 752 & 11 & 40 &ex119 & 84 & 479 & 23 & 16 \\ \hline
ex021 & 318 & 572 & 9 & 2.80 &ex121 & 204 & 1164 & 34 & 76 \\ \hline
ex023 & 690 & 1355 & 8 & 0.91 &ex123 & 122 & 635 & 35 & 14 \\ \hline
ex025 & 92 & 472 & 20 & 1.61 &ex125 & 320 & 8862 & 70 & 8.19 \\ \hline
ex027 & 274 & 715 & 11 & 51 &ex127 & 228 & 527 & 10 & 0.20 \\ \hline
ex029 & 238 & 411 & 9 & 1.33 &ex129 & 737 & 2826 & 14 & 0.97 \\ \hline
ex031 & 219 & 382 & 8 & 12 &ex131 & 292 & 1386 & 18 & 0.17 \\ \hline
ex033 & 363 & 541 & 7 & 50 &ex133 & 522 & 1296 & 11 & 3.94 \\ \hline
ex035 & 247 & 804 & 14 & 3.60 &ex135 & 2822 & 129474 & 87 & 49 \\ \hline
ex037 & 272 & 615 & 10 & 3.43 &ex137 & 196 & 1098 & 19 & 0.34 \\ \hline
ex039 & 56 & 280 & 32 & 58 &ex139 & 334 & 568 & 9 & 8.34 \\ \hline
ex041 & 205 & 341 & 9 & 0.63 &ex141 & 226 & 1168 & 34 & 117 \\ \hline
ex043 & 279 & 513 & 9 & 3.34 &ex143 & 130 & 660 & 35 & 52 \\ \hline
ex045 & 600 & 865 & 7 & 7.80 &ex145 & 48 & 96 & 12 & 18 \\ \hline
ex047 & 1854 & 21118 & 21 & 140 &ex147 & 101 & 606 & 16 & 0.093 \\ \hline
ex049 & 117 & 332 & 13 & 0.078 &ex149 & 698 & 2604 & 12 & 0.75 \\ \hline
ex051 & 136 & 254 & 10 & 0.62 &ex151 & 279 & 733 & 12 & 210 \\ \hline
ex053 & 218 & 383 & 9 & 1.98 &ex153 & 772 & 11654 & 47 & 57 \\ \hline
ex055 & 197 & 813 & 18 & 0.078 &ex155 & 758 & 11580 & 47 & 103 \\ \hline
ex057 & 281 & 9075 & 117 & 0.093 &ex157 & 260 & 467 & 9 & 6.42 \\ \hline
ex059 & 298 & 780 & 10 & 0.47 &ex159 & 582 & 2772 & 18 & 2.37 \\ \hline
ex061 & 158 & 1058 & 22 & 9.59 &ex161 & 1046 & 3906 & 12 & 2.84 \\ \hline
ex063 & 103 & 582 & 34 & 4.76 &ex163 & 244 & 445 & 10 & 4.69 \\ \hline
ex065 & 50 & 175 & 25 & 79 &ex165 & 222 & 742 & 14 & 0.23 \\ \hline
ex067 & 235 & 424 & 10 & 2.70 &ex167 & 509 & 969 & 10 & 7.96 \\ \hline
ex069 & 235 & 441 & 9 & 1.43 &ex169 & 3706 & 42236 & 22 & 530 \\ \hline
ex071 & 253 & 434 & 9 & 2.42 &ex171 & 647 & 2175 & 14 & 0.77 \\ \hline
ex073 & 712 & 1085 & 7 & 15 &ex173 & 536 & 1011 & 10 & 5.05 \\ \hline
ex075 & 111 & 360 & 8 & 0.28 &ex175 & 227 & 1000 & 17 & 113 \\ \hline
ex077 & 237 & 423 & 10 & 2.70 &ex177 & 227 & 759 & 14 & 0.23 \\ \hline
ex079 & 314 & 4943 & 42 & 1.64 &ex179 & 187 & 346 & 10 & 14 \\ \hline
ex081 & 188 & 638 & 6 & 0.55 &ex181 & 109 & 732 & 18 & 0.20 \\ \hline
ex083 & 213 & 380 & 10 & 3.05 &ex183 & 265 & 471 & 11 & 8.61 \\ \hline
ex085 & 229 & 370 & 8 & 11 &ex185 & 237 & 793 & 14 & 0.33 \\ \hline
ex087 & 380 & 5790 & 47 & 46 &ex187 & 240 & 453 & 10 & 2.80 \\ \hline
ex089 & 318 & 576 & 9 & 11 &ex189 & 178 & 4517 & 70 & 3.59 \\ \hline
ex091 & 193 & 336 & 9 & 31 &ex191 & 492 & 1608 & 15 & 21 \\ \hline
ex093 & 454 & 664 & 7 & 27 &ex193 & 1391 & 3012 & 10 & 3.80 \\ \hline
ex095 & 220 & 555 & 11 & 0.59 &ex195 & 216 & 382 & 10 & 6.11 \\ \hline
ex097 & 286 & 4079 & 48 & 2.01 &ex197 & 303 & 1158 & 15 & 0.36 \\ \hline
ex099 & 616 & 923 & 7 & 88 &ex199 & 310 & 537 & 9 & 23 \\ \hline    
\end{tabular}
   }
    \caption{Results on the PACE 2017 public instances}
    \label{tab:PACE2017}
\end{center}
\end{table}

\section*{Acknowledgment}
The author thanks Hiromu Ohtsuka for his help in
implementing the block sieve data structure.
He also thanks Yasuaki Kobayashi for helpful discussions
and especially for drawing the author's attention to
the notion of safe separators. This work would have been non-existent
if not motivated by the timely challenges of
PACE 2016 and 2017. The author is deeply indebted to
their organizers, especially Holger Dell, 
for their dedication and excellent work.

\end{document}